\documentclass[10pt,a4paper]{article}

\usepackage{graphicx}
\usepackage{amsmath} 
\usepackage{amsthm}
\usepackage{amssymb}
\usepackage{array}
\usepackage{makecell}
\usepackage{subfiles} 
\usepackage{adjustbox}
\usepackage{version}
\usepackage{fullpage}
\usepackage{palatino}
\usepackage{authblk}
\usepackage[most]{tcolorbox} 
\usepackage{tikz}
\usetikzlibrary{arrows}
\usetikzlibrary{shapes.geometric}
\usetikzlibrary{decorations.markings}

\usepackage[shortlabels]{enumitem}

\def\probleme#1#2#3{%
\begin{tcolorbox}[enhanced,
  attach boxed title to top left={xshift=2mm,yshift=-3mm,yshifttext=-2mm},
  colback=gray!20!white, colframe=gray!70!black, colbacktitle=gray!70!black,
  boxrule=0.2mm, boxed title style={size=small}, title={#1}]
 \begin{description}[noitemsep,font=\textsf]
  \item[Instance:] {#2}
  \item[Question:] {#3}
 \end{description}
\end{tcolorbox}
}

\newtheorem{theorem}{Theorem}

\newtheorem{lemma}[theorem]{Lemma}

\newtheorem{corollary}[theorem]{Corollary}

\newtheorem{observation}[theorem]{Observation}

\numberwithin{subcase}{case}
\makeatletter
\@addtoreset{case}{theorem}
\makeatother

\newcommand{\haruka}[1]{#1}

\newcommand{\set}[1]{\{#1\}}
\newcommand{\neig}[2]{N_{#1}(#2)} 
\newcommand{\cneig}[2]{N_{#1}[#2]} 
\newcommand{\ini}{{\sf 0}}

\newcommand{\desire}{{\sf t}}
\newcommand{\TAR}{\textsf{TAR}}

\newcommand{\onestep}{\leftrightarrow} 
\newcommand{\sevstep}{\leftrightsquigarrow} 
\newcommand{\sevstepk}[1]{\overset{#1}{\sevstep}} 
\newcommand{\thr}{k} 
\newcommand{\sol}{s} 
\newcommand{\dom}{D} 
\newcommand{\domcore}{C} 
\newcommand{\dege}{d} 
\newcommand{\pw}{pw} 
\newcommand{\mdeg}{\Delta} 
\newcommand{\vc}{\tau} 

\newcommand{\rem}{v_{r}} 
\newcommand{\lea}{v_{l}} 
\newcommand{\Gk}{G_k}
\newcommand{\Dk}{\dom_k}

\date{}

\begin{document}

\title{Decremental Optimization of Dominating Sets Under the Reconfiguration Framework \footnote{Partially supported by JSPS and MAEDI under the Japan-France Integrated Action Program (SAKURA). The first and third author is partially supported by ANR project GrR(ANR-18-CE40-0032). The second author is partially supported by JSPS KAKENHI Grant Number JP19J10042, Japan. The third author is partially supported by ANR project GraphEn (ANR-15-CE40-0009). The fourth author is partially supported by JST CREST Grant Number JPMJCR1402, and JSPS KAKENHI Grant Numbers JP17K12636, JP18H04091 and JP20K11666, Japan.}
}

\author[1]{Alexandre Blanch\'e}
\author[2]{Haruka Mizuta}
\author[1]{Paul Ouvrard}
\author[2]{Akira Suzuki}

\affil[1]{Univ. Bordeaux, Bordeaux INP, CNRS, LaBRI, UMR5800, F-33400 Talence, France \thanks{\{alexandre.blanche, paul.ouvrard\}@u-bordeaux.fr}}
\affil[2]{Graduate School of Information Sciences, Tohoku University, Aoba 6-6-05, Aramaki-aza, Aoba-ku, Sendai, Miyagi, 980-8579, Japan \thanks{haruka.mizuta.s4@dc.tohoku.ac.jp, a.suzuki@ecei.tohoku.ac.jp}}

\maketitle              

\begin{abstract}
Given a dominating set, how much smaller a dominating set can we find through elementary operations? Here, we proceed by iterative vertex addition and removal while maintaining the property that the set forms a dominating set of bounded size. This can be seen as the optimization variant of the dominating set reconfiguration problem, where two dominating sets are given and the question is merely whether they can be reached from one another through elementary operations.
We show that this problem is PSPACE-complete, even if the input graph is a bipartite graph, a split graph, or has bounded pathwidth. On the positive side, we give linear-time algorithms for cographs, trees and interval graphs.
We also study the parameterized complexity of this problem. More precisely, we show that the problem is W[2]-hard when parameterized by the upper bound on the size of an intermediary dominating set. On the other hand, we give fixed-parameter algorithms with respect to the minimum size of a vertex cover, or $d+s$ where $d$ is the degeneracy and $s$ is the upper bound of the output solution.
\end{abstract}

\section{Introduction}
		
	\begin{figure}[t]
		\centering
		\includegraphics[width=0.85\textwidth]{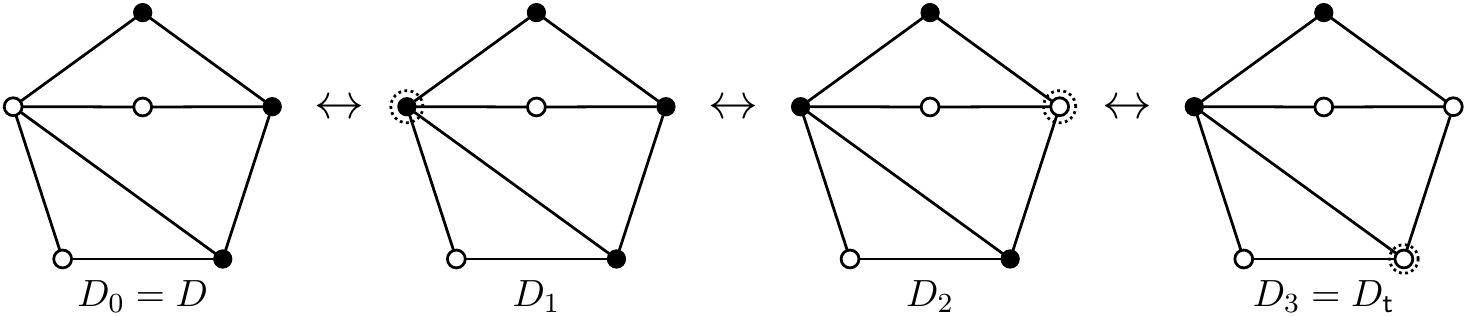}
		\caption{Reconfiguration sequence between $\dom_0$ and $\dom_3$ via dominating sets $\dom_1,\dom_2$ with upper bound $\thr = 4$, where vertices contained in a dominating set are depicted by black circles, and added or removed vertices are surrounded by dotted circles.}
		\label{fig:sequence}
	\end{figure}
	
	Recently, {\em Combinatorial reconfiguration}~\cite{IDHPSUU11} has been extensively studied in the field of theoretical computer science.
	(See, e.g., surveys \cite{H13,N18}.)
	A reconfiguration problem is generally defined as follows:
	we are given two feasible solutions of a combinatorial search problem, and asked to determine whether we can transform one into the other via feasible solutions so that all intermediate solutions are obtained from the previous one by applying the specified reconfiguration rule.
	This framework is applied to several well-studied combinatorial search problems; for example,
	{\sc Independent Set}~\cite{BB17,HD05,KMM12,LM18},
	{\sc Vertex Cover}~\cite{MNR14,MNRSS17},
	{\sc Dominating Set}~\cite{HIMNOST16,LMPRS18,MNRSS17,SMN16}, and so on.
	
	The {\sc Dominating Set Reconfiguration} problem is one of the well-studied reconfiguration problems.
	For a graph $G=(V,E)$, a vertex subset $\dom \subseteq V$ is called a {\em dominating set} of $G$ if $\dom$ contains at least one vertex in the closed neighborhood of each vertex in $V$.
	Figure~\ref{fig:sequence} illustrates four dominating sets of the same graph.
	Suppose that we are given two dominating sets $\dom_\ini$ and $\dom_\desire$ of a graph whose cardinalities are at most a given upper bound $\thr$.
	Then the {\sc Dominating Set Reconfiguration} problem asks to determine whether we can transform $\dom_\ini$ into $\dom_\desire$ via dominating sets of cardinalities at most $\thr$ such that all intermediate ones are obtained from the previous one by adding or removing exactly one vertex.
	Note that this reconfiguration rule, i.e. adding or removing exactly one vertex while keeping the cardinality constraint, is called {\em the token addition and removal} ({\TAR}) rule.
	Figure~\ref{fig:sequence} illustrates an example of transformation between two dominating sets $\dom_0$ and $\dom_3$ for an upper bound $\thr = 4$.

	Combinatorial reconfiguration models ``dynamic'' transformations of systems, where we wish to transform the current configuration of a system into a more desirable one by a step-by-step transformation.
	In the current framework of combinatorial reconfiguration, we need to have in advance a target (a more desirable) configuration.
	However, it is sometimes hard to decide a target configuration, because there may exist exponentially many desirable configurations.
	Based on this situation, Ito \emph{et al.} introduced the new framework of reconfiguration problems, called {\em optimization variant}~\cite{OPT-ISR}.
	In this variant, we are given a single solution as a current configuration, and asked for a more ``desirable'' solution	reachable from the given one.
	This variant was introduced very recently, hence it has only been applied to {\sc Independent Set Reconfiguration} to the best of our knowledge.
    Therefore and since {\sc Dominating Set Reconfiguration} is one of the well-studied reconfiguration problems as we already said, we focus on this problem and study it under this framework.
	
	\subsection{Our problem}
		In this paper, we study the optimization variant of {\sc Dominating Set Reconfiguration}, denoted by {\sc OPT-DSR}. To avoid confusion, we call the original {\sc Dominating Set Reconfiguration} the {\em reachability variant}, and  we denote it by {\sc REACH-DSR}. 
		Suppose that we are given a graph $G$, two integers $\thr,\sol$, and a dominating set $\dom$ of $G$ whose cardinality is at most $\thr$; we call $\thr$ an {\em upper bound} and $\sol$ a {\em solution size}.
		Then {\sc OPT-DSR} asks for a dominating set $\dom_\desire$ satisfying the following two conditions: (a)~the cardinality of $\dom_\desire$ is at most $\sol$, and (b)~$\dom_\desire$ can be transformed from $\dom$ under the {\TAR} rule with upper bound $\thr$. 
		For example, if we are given a dominating set $\dom_0$ in Figure~\ref{fig:sequence} and two integers $\thr = 4$ and $\sol = 2$, then one of the solutions is $\dom_3$, because $\dom_3$ can be transformed from $\dom_0$ and $|D_3| \leq 2$ holds.

	\subsection{Related results}
		Although {\sc OPT-DSR} is being introduced in this paper, some results for {\sc REACH-DSR} relate to {\sc OPT-DSR} in the sense that the techniques to show the computational hardness or construct an algorithm will be used in our proof for {\sc OPT-DSR}.
		We thus list such results for {\sc REACH-DSR} in the following.
		
		There are several results for the polynomial-time solvability of {\sc REACH-DSR}.
		Haddadan \emph{et al.}~\cite{HIMNOST16} showed that {\sc REACH-DSR} under {\TAR} rule is PSPACE-complete for split graphs, for bipartite graphs, and for planar graphs, while linear-time solvable for interval graphs, for cographs, and for forests.
		{\sc REACH-DSR} is also studied well from the viewpoint of fixed-parameter (in)tractability.
		Mouawad \emph{et al.}~\cite{MNRSS17} showed that {\sc REACH-DSR} under {\TAR} is W[2]-hard when parameterized by an upper bound $\thr$.
		As a positive result, Lokshtanov \emph{et al.}~\cite{LMPRS18} gave a fixed-parameter algorithm with respect to $\thr + d$ for graphs that exclude $K_{d,d}$ as a subgraph.
		
		Ito {\em et al.} studied optimization variant of {\sc Independent Set Reconfiguration} ({\sc OPT-ISR})~\cite{OPT-ISR}.
        More precisely, they proved that this problem is PSPACE-hard on bounded pathwidth, NP-hard on planar graphs, while linear-time solvable on chordal graphs.
        They also gave an XP-algorithm with respect to the solution size, and a fixed-parameter algorithm with respect to both solution size and degeneracy.

	\subsection{Our results}
	
		\begin{figure}[bt]
    		\centering
			\includegraphics[width=0.9\textwidth]{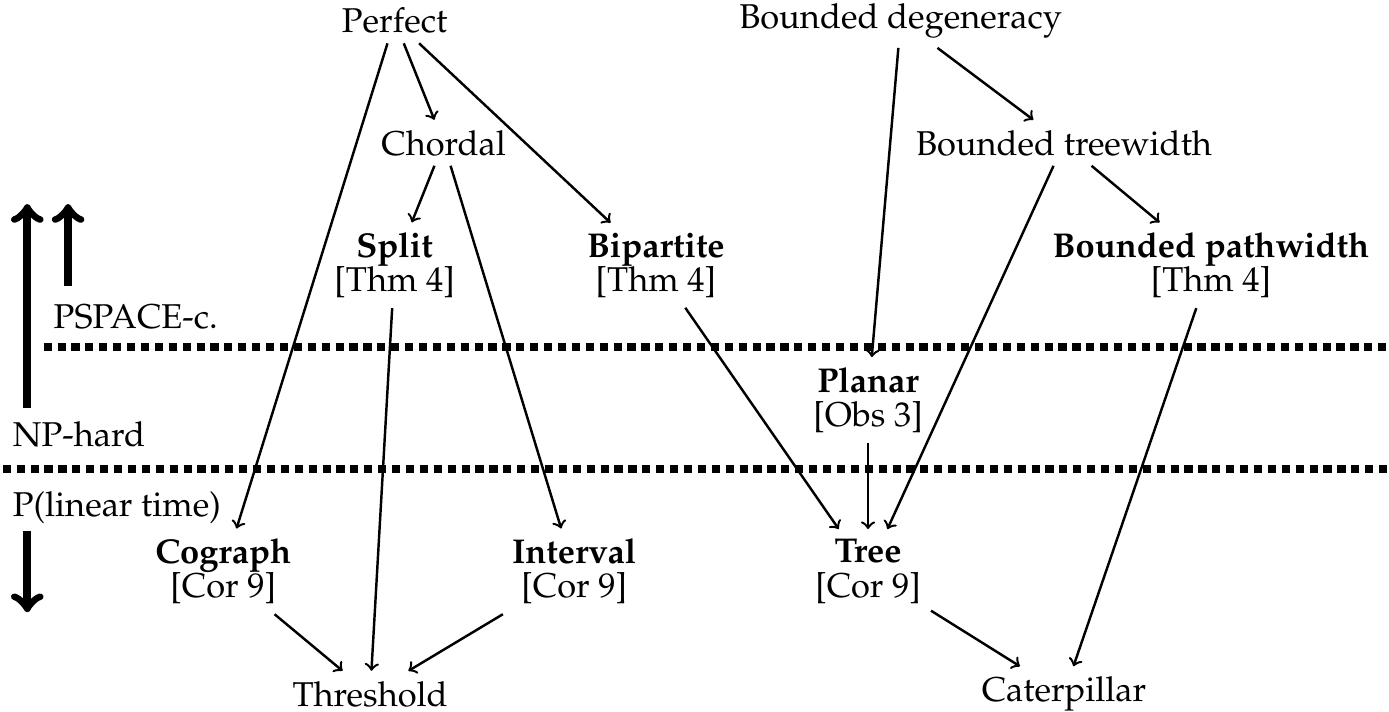}
    		\caption{Our results for polynomial-time solvability with respect to graph classes, where A~$\rightarrow$~B means that the class A contains the class B.} \label{fig:graphclasses}
		\end{figure}

		In this paper, we study {\sc OPT-DSR} from the viewpoint of the polynomial-time (in)tractability and fixed-parameter (in)tractability.
		
		We first study the polynomial-time solvability of {\sc OPT-DSR} with respect to graph classes (See Figure~\ref{fig:graphclasses}).
		Specifically, we show that the problem is PSPACE-complete even for split graphs, for bipartite graphs, and for bounded pathwidth graphs, and NP-hard for planar graphs with bounded maximum degree.
		On the other hand, the problem is linear-time solvable for cographs, trees and interval graphs.
		The inclusions of these graph classes are represented in Figure~\ref{fig:graphclasses}.
		
		We then study the fixed-parameter (in)tractability of {\sc OPT-DSR}.
		We first focus on the following four graph parameters: the degeneracy $\dege$, the maximum degree $\mdeg$, the pathwidth $\pw$, and the vertex cover number $\vc$ (that is the size of a minimum vertex cover).
		Figure~\ref{fig:result_parameter}(a) illustrates the relationship between these parameters, where A~$\rightarrow$~B means that the parameter A is bounded by some function of B.
		This relation implies that if we have a result stating that {\sc OPT-DSR} is fixed-parameter tractable for A then the tractability for B follows, while if we have a negative (i.e. intractability) result for B then it extends to A.
		From results for polynomial-time solvability, we show the PSPACE-completeness for fixed $\pw$ and NP-hardness for fixed $\mdeg$, and hence the problem is fixed-parameter intractable for each parameter $\pw$, $\mdeg$ and $\dege$ under P $\neq$ PSPACE or P $\neq$ NP.
		As a positive result, we give an FPT algorithm for $\vc$.
		We then consider two input parameters: the solution size $\sol$ and the upper bound $\thr$.
		(See Figure~\ref{fig:result_parameter}(c).)
		We show that {\sc OPT-DSR} is W[2]-hard when parameterized by $\thr$.
		We note that we can assume without loss of generality that $\sol < \thr$ holds, as explained in Section~\ref{sec:preliminaries}.
		Therefore, it immediately implies W[2]-hardness for $\sol$.
	 	Most single parameters (except for $\vc$) cause a negative (intractability) result.
		We thus finally consider combinations of one graph parameter and one input parameter. 
		We give an FPT algorithm with respect to $\sol + \dege$.
		(See Figure~\ref{fig:result_parameter}(b).)
		In the end, we can conclude from the discussion above
		that for any combination of a graph parameter $p \in \{ \dege, \mdeg, \pw, \vc \}$ and an input parameter $q \in \{ \sol, \thr \}$, {\sc OPT-DSR} is fixed-parameter tractable when parameterized by $p + q$.
				
		\begin{figure}[bt]
			\centering
			\includegraphics[width=0.9\textwidth]{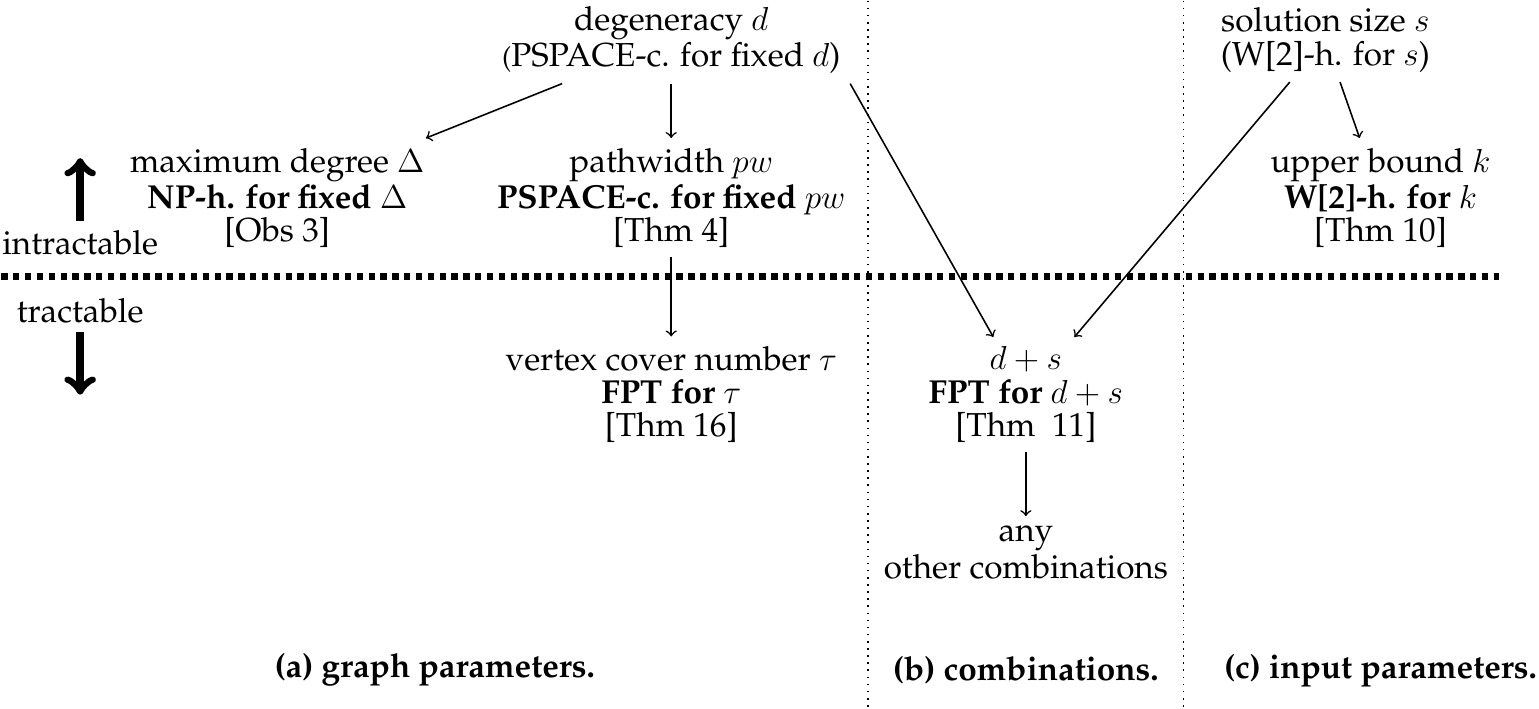}
			\caption{Our results for fixed-parameter tractability, where A~$\rightarrow$~B means that the parameter A is bounded on some function of B.} \label{fig:result_parameter}
		\end{figure}				
		
\section{Preliminaries}\label{sec:preliminaries}
	For a graph $G$, we denote by $V(G)$ and $E(G)$ the vertex set of $G$ and edge set of $G$, respectively.
	For a vertex $v \in V(G)$, we let $\neig{G}{v} = \{ w \mid vw \in E(G)\}$ and $\cneig{G}{v} = \neig{G}{v} \cup \{ v \}$; we call a vertex in $\neig{G}{v}$ a {\em neighbor} of $v$ in $G$.
	For a vertex subset $S \subseteq V(G)$, we let $\cneig{G}{S} = \bigcup_{v \in S}\cneig{G}{v}$.
	If there is no confusion, we sometimes omit $G$ from the notation.
	
	\subsection{Optimization variant of Dominating Set Reconfiguration}
	    For a graph $G=(V,E)$, a vertex subset $\dom \subseteq V$ is a {\em dominating set} of $G$ if $N[\dom] = V(G)$.
	    For a dominating set $\dom$, we say that $u \in \dom$ {\em dominates} $v \in V$ if $v \in N[u]$ holds.
	    We say that a vertex $v \in \dom$ has a \emph{private neighbor} in $\dom$ if there exists a vertex $u \in N[v]$ such that $N[u] \cap \dom = \{v\}$.
	    In other words, the vertex $u$ is dominated only by $v$ in $D$.
	    Note that the private neighbor of a vertex can be itself.
	    A dominating set is (inclusion-wise) {\em minimal} if and only if each of its vertices has a private neighbor, and {\em minimum} if and only if the cardinality is minimum among all dominating sets.
	    Notice that any minimum dominating set is minimal.

	    Let $\dom$ and $\dom^\prime$ be two dominating sets of $G$. We say that $\dom$ and $\dom^\prime$ are {\em adjacent} if $|\dom \Delta \dom^\prime| = 1$, where $\dom \Delta \dom^\prime = (\dom \setminus \dom^\prime) \cup (\dom^\prime \setminus \dom)$ and we denote this by $\dom \onestep \dom^\prime$.
	    Let us now assume that $\dom$ and $\dom^\prime$ are both of size at most $\thr$, for some given $\thr \ge 0$. Then, a {\em reconfiguration sequence} between $D$ and $D^\prime$ under the {\TAR} rule (or sometimes called a {\TAR}-sequence) is a sequence $\langle \dom = \dom_0,\dom_1,\ldots,\dom_\ell = \dom^\prime \rangle$ of dominating sets of $G$ such that:
	   	\begin{itemize}
	    	\item for each $i \in \{ 0,1,\ldots,\ell \}$, $\dom_i$ is a dominating set of $G$ such that $|\dom_i| \leq \thr$; and
	    	\item for each $i \in \{ 0,1,\ldots,\ell-1 \}$, $\dom_i \onestep \dom_{i+1}$ holds.
	    \end{itemize}
	    Considering a reconfiguration sequence under the {\TAR} rule, we sometimes write {\TAR}($\thr$) instead of {\TAR} to emphasize the upper bound $\thr$ on the size of a solution.
	    We say that $\dom^\prime$ is {\em reachable from} $\dom$ if there exists a reconfiguration sequence between $\dom$ and $\dom^\prime$; since a reconfiguration sequence is reversible, if $\dom^\prime$ is reachable from $\dom$, then $\dom$ is also reachable from $\dom^\prime$.
	    We write $\dom \sevstepk{\thr} \dom^\prime$ if $\dom^\prime$ (resp. $\dom$) is reachable from $\dom$ (resp. $\dom^\prime$).
	    Then, the {\em optimization variant} of the {\sc Dominating Set Reconfiguration} problem (OPT-DSR) is defined as follows:
	   
	    \probleme{\textsc OPT-DSR}%
            {A graph $G$, two integers $k,s \ge 0$, a dominating set $D$ of $G$ such that $|D| \le k$.}%
            {A dominating set $D_\desire$ of $G$ such that $|D_\desire| \le s$ and $D \sevstepk{k} D_\desire$ if it exists, {\sf no}-instance otherwise.}
	    
	   We denote by a 4-tuple $(G,\thr,\sol,\dom)$ an instance of OPT-DSR.
    \subsection{Useful observations}
	    From the definition of OPT-DSR, we have the following observations.
	    
	    \begin{observation}\label{obs:inequality}
	    	Let $(G,\thr,\sol,\dom)$ be an instance of OPT-DSR.
	    	If $\thr,\sol$ and $|\dom|$ violate the inequality $\sol < |\dom| \leq \thr$, then $\dom$ is a solution of the instance.
	    \end{observation}
	    
    	\begin{proof}
    		By the definition of $\dom$, we know $|\dom| \leq \thr$.
    		Therefore if the inequality is violated, we have $|\dom| \leq \sol \leq \thr$ or $|\dom| \leq \thr \leq \sol$.
    		In both cases, $|\dom| \leq \sol$ holds, and hence $\dom$ is a solution. 
    	\end{proof}

    	It is observed that the condition in Observation~\ref{obs:inequality} can be checked in linear time.
    	Therefore, we sometimes assume without loss of generality that $\sol < |\dom| \leq \thr$ holds.
    	Then, another observation follows.
    	
    	\begin{observation}\label{obs:minimal}
    		Let $(G,\thr,\sol,\dom)$ be an instance of OPT-DSR such that $\sol < |\dom|$ holds.
    		If $\dom$ is minimal and $|\dom| = \thr$ holds, then the instance has no solution.
    	\end{observation}
    	
    	\begin{proof}
    		Since $|\dom| = \thr$, we cannot add any vertex to $\dom$ without exceeding the threshold $\thr$.
    		Besides, since $\dom$ is minimal, we cannot remove any vertex while maintaining the domination property.
    		As a result, there is no dominating set $\dom_\desire$ of size at most $\sol$ reachable from $\dom$, i.e. $\dom \sevstepk{\thr} \dom_\desire$ does not hold for any dominating set $\dom_\desire$ such that $|\dom_\desire| \le \sol$.
    	\end{proof}  
    	
	    Again, the conditions in Observation~\ref{obs:minimal} can be checked in linear time, and hence we can assume without loss of generality that $\dom$ is not minimal or $|\dom| < \thr$ holds.
	    Suppose that $\dom$ is not minimal.
        Then we can always obtain a dominating set of size less than $\thr$ by removing some vertex without private neighbor from $\dom$, that is, we have a dominating set $\dom^\prime$ with $\dom \sevstepk{k} \dom^\prime$ and $|\dom^\prime| < \thr$.
        Note that $(G,\thr,\sol,\dom)$ has a solution if and only if $(G,\thr,\sol,\dom^\prime)$ does.
	    Therefore, it suffices to consider the case where $|\dom| < \thr$ holds.
	    Combining it with Observation~\ref{obs:inequality}, we sometimes assume without loss of generality that $\sol < |\dom| < \thr$ holds.
	    
	    Finally, we have the following observation which states that {\sc OPT-DSR} is a generalization of the {\sc Dominating Set Problem}:
	    
	    \begin{observation}\label{obs:equivalent}
			Let $G = (V,E)$ be a graph and $s$ be an integer.
			Then the instance $(G,|V|,s,V)$ of OPT-DSR is equivalent to finding a dominating set of $G$ of size at most $s$.	
		\end{observation}
		
		\begin{proof}
		        Let $\dom_\desire$ be a dominating set of $G$ of size at most $s$. Since we started from a dominating set containing all the vertices of $G$, it is sufficient to remove one by one each vertex in $V \setminus \dom_\desire$ to reach $\dom_\desire.$
		\end{proof}
	
		Observation~\ref{obs:equivalent} implies that hardness results for the original \textsc{Dominating Set} problem extend to \textsc{OPT-DSR}. In particular, we get that \textsc{OPT-DSR} is NP-hard even for the case where the input graph has maximum degree 3, or is planar with maximum degree 4~\cite{GJ79}. However, we will show in Section~\ref{sec:PSPACE} that this problem is actually PSPACE-complete.

\section{Polynomial-time (in)tractability}\label{sec:polytime}
    \subsection{PSPACE-completeness for several graph classes}\label{sec:PSPACE}
		\begin{theorem}\label{thm:PSPACE}
			\textsc{OPT-DSR} is PSPACE-complete even when restricted to bounded pathwidth graphs, for split graphs, and for bipartite graphs.
		\end{theorem}
		
		First, observe that OPT-DSR is in PSPACE. Indeed, when we are given a dominating set $D_\desire$ as a solution for some instance of OPT-DSR, we can check in polynomial time whether it has size at most $\sol$ or not. Furthermore, since REACH-DSR is in PSPACE, we can check in polynomial space whether it is reachable from the original dominating set $D$. Therefore, we can conclude that OPT-DSR is in PSPACE.
 
        We now give three reductions to show the PSPACE-hardness for split graphs, bipartite graphs and bounded pathwidth graphs, respectively. These reductions are slight adaptations of the ones of PSPACE-hardness for REACH-DSR developed in~\cite{HIMNOST16}.
        We only give the hardness proof for split graphs; the two other proofs have been moved to Appendix due to space limitation. 
        To this end, we use a polynomial-time reduction from the optimization variant of \textsc{Vertex Cover Reconfiguration}, denoted by OPT-VCR.
        
        Given a graph $G=(V,E)$, a {\em vertex cover} is a subset of vertices that contains at least one endpoint of each edge in $E$.
        We now give the formal definition of OPT-VCR. Suppose that we are given a graph $G$, two integers $\thr,\sol \geq 0$, and a vertex cover $C$ of $G$ whose cardinality is at most $\thr$. 
        Then OPT-VCR asks for a vertex cover $C_\desire$ of size at most $\sol$ reachable from $C$ under the \TAR($\thr$) rule.
        This problem is known to be PSPACE-complete even for bounded pathwidth graphs\footnote{In~\cite{OPT-ISR}, Ito \emph{et al.} actually showed the PSPACE-completeness for the optimization variant of {\sc Independent Set Reconfiguration}. However, the result can easily be converted to OPT-VCR from the observation that any vertex cover of a graph is the complement of an independent set.}~\cite{OPT-ISR}.
        
        \begin{lemma}\label{lem:PSPACEsplit}
        	\textsc{OPT-DSR} is PSPACE-hard even for split graphs.
        \end{lemma}
        
        \begin{proof}
	        \begin{figure}[bt]
		        \centering
		        \includegraphics[width=0.85\textwidth]{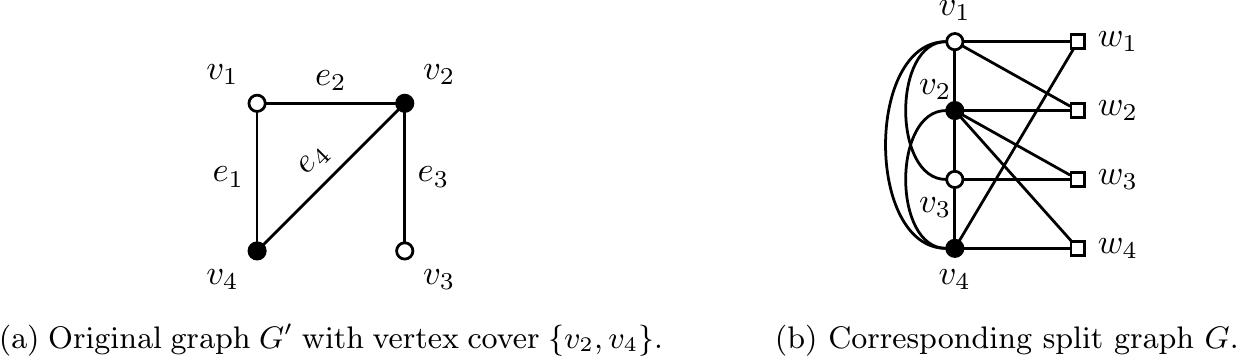}
		         \caption{Reduction for Lemma ~\ref{lem:PSPACEsplit}. Note that $\{v_2,v_4\}$ is a dominating set of $G$.}
		    \label{fig:reduction:split}
	    \end{figure}
	
    	As we said, we give a polynomial-time reduction from \textsc{OPT-VCR}. More precisely, we extend the idea developed for the NP-hardness proof of \textsc{Dominating Set} problem on split graphs~\cite{B84}.
    	
        Let $(G^\prime, \thr^\prime, \sol^\prime, C)$ be an instance of \textsc{OPT-VCR} with vertex set $V(G^\prime) = \{ v_1, v_2, \dots , v_n \}$ and edge set $E(G^\prime) = \{ e_1, e_2, \dots , e_m \}$.
    	We construct the corresponding split graph $G$ as follows (see also Figure ~\ref{fig:reduction:split}).
    	Let $V(G) = A \cup B$, where $A=V(G^\prime)$ and $B= \{ w_1, w_2, \dots , w_m \}$; the vertex $w_i \in B$ corresponds to the edge $e_i \in E(G^\prime)$.
    	We join all pairs of vertices in $A$ so that $A$ forms a clique in $G$.
    	In addition, for each edge $e_i = v_p v_q$ in $E(G^\prime)$, we join $w_i \in B$ with each of $v_p$ and $v_q$.
    	Let $G$ be the resulting graph, and let $(G, \thr = \thr^\prime, \sol = \sol^\prime, \dom = C)$ be the corresponding instance of \textsc{OPT-DSR} (we will prove later that $\dom$ is a dominating set of $G$).
    	Clearly, this instance can be constructed in polynomial time. It remains to prove that $(G', k', s', C)$ is a \textsf{yes}-instance if and only if $(G,k,s,D)$ is a \textsf{yes}-instance.
    	
    	\medskip
    	($\Rightarrow$)
    	We start by the only-if direction. Suppose that $(G',k',s',C)$ is a \textsf{yes}-instance. Then, there exists a vertex cover $C_\desire$ of size at most $s'$ reachable from $C$ under the \TAR($k$') rule. 
    	Since $k'=k$, $s=s'$ and both problems employ the same reconfiguration rule, it suffices to prove that any vertex cover of $G'$ is a dominating set of $G$. 
    	Since $C \subseteq V(G^\prime) = A$ and $A$ is a clique, all vertices in $A \setminus C$ are dominated by the vertices in $C$.
    	Thus, consider a vertex $w_i \in B$, which corresponds to the edge $e_i = v_p v_q$ in $E(G^\prime)$.
    	Then, since $C$ is a vertex cover of $G^\prime$, at least one of $v_p$ and $v_q$ must be contained in $C$.
    	This means that $w_i$ is dominated by the endpoint $v_p$ or $v_q$ in $G$.
    	Therefore, each vertex cover in the reconfiguration sequence between $C$ and $C_\desire$ is a dominating set of $G$ (including $D=C$ and $D_\desire=C_\desire$) and thus, $(G,k,s,D)$ is a \textsf{yes}-instance. 
  
        \medskip
        ($\Leftarrow$)
    	We now focus on the if direction. Suppose that $(G,k,s,D)$ is a \textsf{yes}-instance. Then, there exists a dominating set $D_\desire$ of $G$ of size at most $s$ reachable under the \TAR($k$) rule  by a sequence $\mathcal{R} = \langle D_0,D_1,\dots,D_\desire \rangle$, with $D=D_0$.
    	Recall that $D=C$ and thus $D$ is a vertex cover of $G'$.
        We want to produce a sequence of dominating sets that are subsets of $A$. To this end, we proceed by eliminating the vertices of $B$ that appears in $\mathcal{R}$ one by one from the sequence. Let $i$ be the smallest index such that $D_i \in \mathcal{R}$ contains a vertex $w \in B$ associated with the edge $v_av_b \in E(G)$. Let $j \ge i$ be the largest index such that every dominating set $D_l \in \mathcal{R}$ ($i \le l \le j$) contains $w$. Now we show that $D_{j+1}$ is reachable from $D_{i-1}$ under \TAR($k$) rule without touching $w$, that is, there is a sequence where each dominating set on the sequence does not contain $w$. For every $D_l \in \mathcal{R}$ ($i \le l \le j$) we instead consider the set $D'_l = (D_l \setminus w) \cup \{v_a\}$. Note that $v_a \in N_G(w)$, and $|D'_l| \le |D_l| \le k$. Observe that each $D'_l$ is a dominating set since $N_G[w] \subseteq N_G[v_a]$. If $v_a \in D_{i-1}$, observe that $D_{i-1} = D_i'$. Otherwise, $D_i'$ is obtainable from $D_{i-1}$ in one step since we just replace the addition of $w$ by the one of $v_a$. Moreover, due to the choice of $j$, $D_{j+1} = D_j \setminus \{w\}$. Hence, $D_{j+1}$ contains a vertex in $A$ adjacent to $w$. If this vertex is $v_a$, $D_j' = D_{j+1}$. Otherwise, $D_{j+1} = D_j' \setminus \{v_a\}$, which corresponds to a valid {\sf TAR} move. Finally, since we ensure that each dominating set $D_l'$ with $i \le l \le j$ contains $v_a$, we can ignore each move in the subsequence of $\mathcal{R}$ that touches $v_a$. Hence, either $D_l' = D'_{l+1}$ or $D_l' \onestep D'_{l+1}$ holds, for every $i \le l < j$.
        By ignoring duplicates from the sequence $\langle D_{i-1}, D'_i, \ldots , D'_j, D_{j+1} \rangle$, we obtain a desired subsequence which does not touch $w$. Therefore, we can eliminate $w$ in the subsequence $\langle D_{i-1}, D_i, \ldots , D_j, D_{j+1} \rangle$ of $\mathcal{R}$ by replacing it with the desired subsequence. Hence by repeating this process for each subsequence containing $w$ we get a new sequence that does not touch $w$ at all. We then repeat this process for every vertex of $B$ that appears in $\mathcal{R}$ and we obtain a sequence $\mathcal{R'}$ where each dominating set is a subset of $A$.
    	Finally, observe that any dominating set $\dom$ of $G$ such that $D \subseteq A = V(G^\prime)$ forms a vertex cover of $G^\prime$, because each vertex $w_i \in B$ is dominated by at least one vertex in $D \subseteq V(G^\prime)$.
    	Therefore, $(G',s',k',C)$ is a \textsf{yes}-instance.
    \end{proof}

    Finally, the two following lemmas complete the proof of Theorem \ref{thm:PSPACE}.

    \begin{lemma}\label{lem:PSPACEpathwidth}
    	\textsc{OPT-DSR} is PSPACE-hard even for bounded pathwidth graphs.
    \end{lemma}   
    
    The pathwidth of a graph is defined as follows.
    A \emph{path decomposition} of $G$ is a sequence $\mathcal{P} = (X_1, X_2, \ldots, X_\ell)$, where each $X_i \subseteq V$, for each $i \in \{1,2,\cdots, \ell\}$, satisfies the following properties:
    \begin{enumerate}[(i)]
    	\item each vertex $v \in V$ is contained in (at least) one bag $X_i$;
    	\item each edge $uv \in E$ is contained in (at least) one bag, i.e.\ there exists $X_i$ such that $u,v \in X_i$;
    	\item for every three indices $i \le j \le k$, $X_i \cap X_k \subseteq X_j$.
    \end{enumerate}
    The \emph{width} of a given path decomposition is one less than the size of its largest bag, that is $\max_{1 \le i \le \ell} |X_i|-1$. 
    Finally, the \emph{pathwidth} of $G$, denoted by $pw(G)$, is the minimum width of any path decomposition of $G$. 
    Then the following lemma completes the proof of PSPACE-completeness for bounded pathwidth graphs.
  
    \begin{proof}
    	Our reduction follows from the original reduction from \textsc{Vertex Cover} to \textsc{Dominating Set}~\cite{GJ79}.
    	Let $(G^\prime, \thr^\prime, \sol^\prime, C)$ be an instance of \textsc{OPT-VCR}, \haruka{where the pathwidth of $G^\prime$ is bounded; note again that \textsc{OPT-VCR} is PSPACE-complete even for bounded pathwidth graphs.}
    	Let $G$ be the graph constructed from $G^\prime$ as follows: for each edge $u,w$, we add a new vertex $v_{uw}$ and join it with both of $u$ and $w$ by edges (see Fig.~\ref{fig:reduction:pathwidth}).
    	\haruka{We now claim that the pathwidth of $G$ is bounded.
    	Let $\mathcal{P}^\prime$ be the path decomposition of $G^\prime$ of width $pw(G^\prime)$.
    	Then we construct a path decomposition $\mathcal{P}$ of $G$ from $\mathcal{P}^\prime$ by adding each new vertex $v_{uw}$ to any vertex subset $X^\prime_i$ in $\mathcal{P}^\prime$ in which both $u$ and $w$ are contained; from the definition of a path decomposition, such a vertex subset $X^\prime_i$ always exists.
        For each vertex subset in $\mathcal{P}^\prime$, the number of pairs of two vertices is $O(\pw(G^\prime)^2)$, and hence the resulting path decomposition $\mathcal{P}$ has width $O(\pw(G^\prime)^2)$; it is bounded by some constant since $\pw(G^\prime)$ is.}
    	Let $(G, \thr = \thr^\prime, \sol = \sol^\prime, \dom = C)$ be the corresponding instance of \textsc{OPT-DSR}.
    	This construction can clearly be done in polynomial time.
    
    	\begin{figure}[bt]
    		\centering
        	\includegraphics[width=0.9\textwidth]{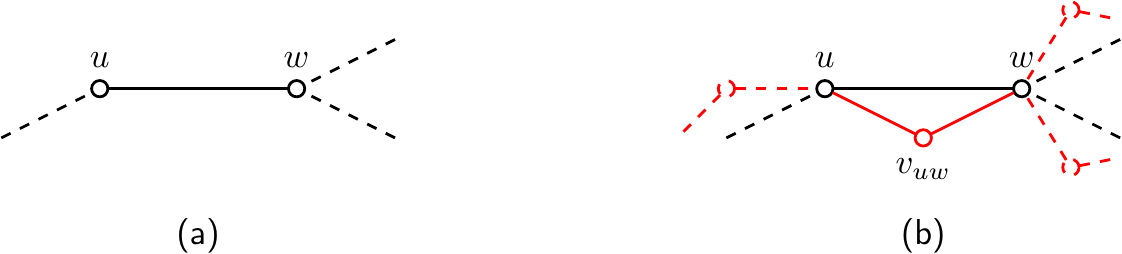}
        	\caption{Reduction for Lemma~\ref{lem:PSPACEpathwidth}. {\sf (a)} Original edge $uv$ in $G'$ and {\sf (b)} gadget in $G$ for $uv$.}
        	\label{fig:reduction:pathwidth}
        \end{figure}
        	
        It remains to prove that $(G',k',s',C)$ is a \textsf{yes}-instance for \textsc{OPT-VCR} if and only if $(G,k,s,D)$ is a \textsf{yes}-instance for \textsc{OPT-DSR}.
        
        \medskip
        ($\Rightarrow$)
        Suppose that $(G', k', s', C)$ is a \textsf{yes}-instance and let $C_\desire$ be a vertex cover of size at most $s'$ reachable from $C$ under the \TAR($\thr'$) rule, by a sequence $\mathcal{R^\prime}$.
        Since any vertex cover of $G'$ is a dominating set of $G$ and $k=k', s=s'$, then the sequence $\mathcal{R^\prime}$ yields a reconfiguration sequence from $D=C$ to $D_\desire=C_\desire$. Thus, $(G,k,s,D)$ is a \textsf{yes}-instance.
        		
        \medskip		
        ($\Leftarrow$)		
        We now prove the other direction. Suppose that $(G,k,s,D)$ is a \textsf{yes}-instance and let $\mathcal{R} = \langle D_0,D_1,\dots,D_\desire\rangle$ be a \TAR($k$) sequence of dominating sets of $G$ starting at $D = D_0$ and reaching a dominating set $D_\desire$ that satisfies $\vert D_\desire\vert \leq s$.				
        Recall that $\dom$ does not contain any newly added vertex in $V(G) \setminus V(G^\prime)$.
        We want a sequence $\mathcal{R'}$ that does not touch any newly added vertex $v_{uw}$. 
        To this end, in the same spirit as in the proof of Lemma~\ref{lem:PSPACEsplit}, we eliminate the vertices of $V(G) \setminus V(G')$ one by one. If a $D_i$ contains a vertex $v_{uw}$, then we replace $D_i$ by $D_i' = (D_i \setminus v_{uw}) \cup \{u\}$, which is also a dominating set and is reachable in one step from $D_{i-1}$. Thus, the resulting sequence does not touch $v_{uw}$, and by repeating the operation to all vertices of $V(G) \setminus V(G')$, we obtain the wanted \TAR($k$) sequence $\mathcal{R}'$ of subsets of $V(G')$.
        In this way, we can obtain a reconfiguration sequence of vertex covers in $G^\prime$ between $C$ and $C_\desire = D_\desire$ as needed.
        
        Since \textsc{OPT-VCR} is PSPACE-complete for bounded pathwidth graphs, the reduction above implies PSPACE-hardness on bounded pathwidth graphs.
    \end{proof}
    
    \begin{lemma}\label{lem:PSPACEbipartite}
        \textsc{OPT-DSR} is PSPACE-hard even for bipartite graphs.
    \end{lemma}
    
    \begin{proof}
	We give a polynomial-time reduction from \textsc{OPT-DSR} on split graphs to the same problem restricted to bipartite graphs.
	The same idea is used in the NP-hardness proof of \textsc{Dominating Set} problem on bipartite graphs~\cite{B84}.
	
	Let $(G^\prime , \thr^\prime, \sol^\prime, \dom^\prime )$ be an instance of \textsc{OPT-DSR}, where $G^\prime$ is a split graph.
	Then $V(G^\prime)$ can be partitioned into two subsets $A$ and $B$ which form a clique and an independent set in $G^\prime$, respectively.
	Furthermore, by the reduction given in the proof of Lemma~\ref{lem:PSPACEsplit}, the problem on split graph remains PSPACE-complete even if the given dominating set $\dom^\prime$ consists of vertices only in $A$.
	We thus assume that $\dom^\prime \subseteq A$ holds.
	
	We now construct the corresponding bipartite graph $G$, as follows.
	First, we delete any edge joining two vertices in $A$ so that $A$ forms an independent set.
	Then, we add a new edge consisting of two new vertices $x$ and $y$, and join $y$ with each vertex in $A$.
	
    \begin{figure}[bt]
	    \centering
		\includegraphics[width=0.85\textwidth]{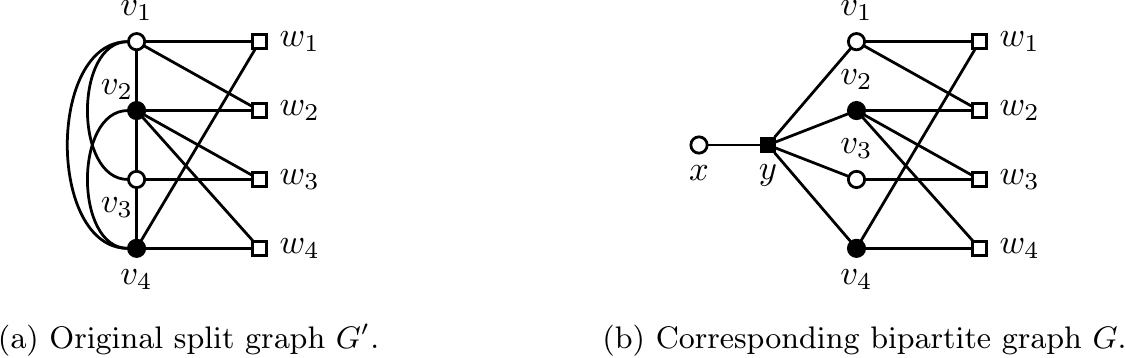}
		\caption{Reduction for Lemma ~\ref{lem:PSPACEbipartite}. Note that $\{v_2,v_4\}$ forms a dominating set of $G'$, and $\{y,v_2,v_4\}$ a dominating set of $G$.}
		\label{fig:reduction:bipartite}
	\end{figure}
		    
	The resulting graph $G$ is bipartite (see Figure~\ref{fig:reduction:bipartite} for an example).
	Let $\dom = \dom^\prime \cup \{ y \}$, $\thr = \thr^\prime + 1 $ and $\sol = \sol^\prime + 1 $.
	Then we obtain the corresponding \textsc{OPT-DSR} instance $(G , \thr, \sol, \dom )$ where $G$ is bipartite (here again, we will prove later that $\dom$ is dominating set of $G$).
	Clearly, this instance can be constructed in polynomial time.
	We then prove that $(G',k',s',D')$ is a \textsf{yes}-instance if and only if $(G,k,s,D)$ is a \textsf{yes}-instance.

    \medskip
    ($\Rightarrow$)
	We first prove the only-if direction.
	Suppose that there exists a dominating set $\dom_\desire^\prime$ of $G^\prime$ such that $\dom^\prime \sevstepk{\thr^\prime} \dom_\desire^\prime$ and $|\dom_\desire^\prime| \leq \sol^\prime$.
	Consider any dominating set $\dom^{\prime\prime}$ of $G^\prime$.
	Then, $B \subseteq \cneig{G}{\dom^{\prime\prime}}$ holds because $B \subseteq \cneig{G^\prime}{\dom^{\prime\prime}}$ and we have deleted only the edges which have both endpoints in $A$.
	Since $\cneig{G}{y} = A \cup \{ x \}$, we can conclude that $\dom^{\prime\prime} \cup \{ y \}$ is a dominating set of $G$.
	Furthermore, $| \dom_\desire^\prime \cup \{ y \} | \le \sol^\prime + 1 = \sol$.
	Thus there exists a dominating set $\dom_s$ of $G$ such that $\dom \sevstepk{\thr} \dom_\desire$ and $|\dom_\desire| \leq \sol$, as desired.
	
	\medskip
	($\Leftarrow$)
	We then prove the if direction.
	Suppose that there exists a dominating set $\dom_\desire$ of $G$, of size at most $s$ and reachable from $\dom$ by a \TAR($k$) sequence $\mathcal{R} = \langle D_0,D_1,\dots,D_\desire \rangle$, with $D=D_0$.
	Recall that $\dom = \dom^\prime \cup \{ y \}$, and notice that any dominating set of $G$ contains at least one of $x$ and $y$.
	Since $\cneig{G}{x} \subset \cneig{G}{y}$, we can assume that $\dom_s$ contains $y$.
	Therefore, we can also assume that $y$ is contained in every dominating set of the reconfiguration sequence.
	Recall that the assumption $\dom^\prime \subseteq A$ holds.
	As in the proofs of Lemmas~\ref{lem:PSPACEsplit} and \ref{lem:PSPACEpathwidth}, we can produce an equivalent sequence $\mathcal{R'}$ that does not touch any vertex of $B$. Again, if a dominating set $D_i$ touches a vertex $w_j$ associated with the edge $v_k,v_l$, we replace $D_i$ by $D_i' = (D_i \setminus w_j) \cup v_k$. We repeat the operation for all $w_j$ and obtain the wanted sequence.
	
	Consider any dominating set $D$ of $G$ in such a reconfiguration sequence.
	Since $y \in \dom$, we have $| D \cap V(G^\prime) | \le \thr -1 = \thr^\prime$.
	Furthermore, since $D \cap V(G^\prime) \subseteq A$ and $A$ forms a clique in $G^\prime$, we have $A \subseteq \cneig{G^\prime}{\dom \cap V(G^\prime)}$.
	Since there is no edge joining $y$ and a vertex in $B$, each vertex in $B$ is dominated by some vertex in $D \cap V(G^\prime)$.
	Therefore, $D \cap V(G^\prime)$ is a dominating set of $G^\prime$ with cardinality at most $\thr^\prime$, and hence there exists a dominating set $\dom_\desire^\prime$ of $G^\prime$ such that $\dom^\prime \sevstepk{\thr^\prime} \dom_\desire^\prime$ and $|\dom_\desire^\prime| \leq \sol^\prime$.
\end{proof}
    
	\subsection{Linear-time algorithms} \label{sec:linear}	
		We now explain how OPT-DSR can be solved in linear time for several graph classes.
		To this end, we deal with the concept of a canonical dominating set. A dominating set $\dom_{\sf c}$ is {\em canonical} if $\dom_{\sf c}$ is a minimum dominating set which is reachable from any dominating set $\dom$ under the \TAR($|\dom|+1$) rule.
		Then we have the following theorem.
		
		\begin{theorem}\label{thm:ptime}
			Let $\mathcal{G}$ be a class of graphs such that any graph $G \in \mathcal{G}$ has a canonical dominating set and we can compute it in linear time.
			Then OPT-DSR can be solved in linear time on $\mathcal{G}$.
		\end{theorem}
		
		\begin{proof}
        	Let $(G,\thr,\sol,\dom)$ be an instance of OPT-DSR, where $G \in \mathcal{G}$.
	        Recall that we can assume without loss of generality that $\sol < |\dom| < \thr$; we can check in linear time whether the inequality is satisfied or not, and if it is violated, then we know from Observation~\ref{obs:inequality}~and~\ref{obs:minimal} that it is a trivial instance.
	        Since $G \in \mathcal{G}$, $G$ admits a canonical dominating set and we can compute in linear time an actual one. Let $\dom_{\sf c}$ be such a canonical dominating set.
	        Then it follows from the definition that $\dom_{\sf c}$ is reachable from $\dom$ under the {\TAR}($\thr$) rule since $\thr \geq |\dom| +1$.
	        Since $\dom_{\sf c}$ is a minimum dominating set, we can output it if $|\dom_{\sf c}| \leq \sol$ holds, and {\sf no}-instance otherwise.
	        All processes can be done in linear time, and hence the theorem follows.
        \end{proof}	

		Haddadan {\em et al.} showed in~\cite{HIMNOST16} that cographs, trees (actually, forests), and interval graphs admit a canonical dominating set.
		Their proofs are constructive, and hence we can find an actual canonical dominating set.
		It is observed that the constructions on cographs and trees can be done in linear time.
		The construction on interval graphs can also be done in linear time with a nontrivial adaptation by using an appropriate data structure.
		Therefore, we have the following linear-time solvability of OPT-DSR.
		\begin{corollary}\label{cor:linear}
			OPT-DSR can be solved in linear time on cographs, trees, and interval graphs.
		\end{corollary}
		
\section{Fixed-parameter (in)tractability}
	In this section, we study the fixed-parameter complexity of OPT-DSR with respect to several graph parameters: the upper bound $k$, solution size $s$, minimum size of a vertex cover $\tau$ and degeneracy $d$. 
	
	More precisely, we first show that OPT-DSR is W[2]-hard when parameterized by the upper bound $k$. To prove it, we use the idea of the reduction constructed by Mouawad \emph{et al.} in~\cite{MNRSS17} to show the W[2]-hardness of REACH-DSR.
	
	\begin{theorem}\label{the:w2-h}
		OPT-DSR is W[2]-hard when parameterized by the upper bound $\thr$.
	\end{theorem}
    
    \begin{proof}
	We give an FPT-reduction from the (original) {\sc Dominating Set} problem that is W[2]-hard when parameterized by its natural parameter $k$ \cite{DF19}.

	Let $(G^\prime,\thr^\prime)$ be an instance of {\sc Dominating Set}, where $|V(G^\prime)| = n^\prime$ and $V(G^\prime) = \{ v_1,v_2,\ldots,v_{n^\prime} \}$.
	Then we construct the corresponding instance $(G,\thr,\sol,\dom)$ of OPT-DSR, as follows.
	We first describe the construction of $G$.
	Let $G_0$ be the graph obtained by adding a universal vertex $v_0$ to $G^\prime$, and $G_1,G_2,\ldots,G_{\thr^\prime}$ be $\thr^\prime$ copies of $G_0$.
	The vertex set of $G$ consists of $\bigcup_{j \in \{ 0,1,\ldots,\thr^\prime \}}V(G_j)$.
	For any $j \in \{1,2,\ldots,\thr^\prime \}$ and $i \in \{0,1,\ldots,n^\prime \}$, we use $v_{j,i}$ to denote the vertex in $G_j$ corresponding to $v_i$ in $G_0$.
	Then, for each vertex $v_i$ in $G_0$ except for $v_0$, we connect $v_i$ by new edges to all vertices in $\cneig{G_j}{v_{j,i}}$ in each $j \in \{ 1,2,\ldots,\thr^\prime \}$; formally, the edge set of $G$ consists of $\bigcup_{j \in \{0,1,\ldots,\thr^\prime \}} E(G_j) \cup \bigcup_{i \in \{ 1,2,\ldots,n^\prime \}}\bigcup_{j \in \{ 1,2,\ldots,\thr^\prime \}}\{ v_{i}w \mid w \in \cneig{G_j}{v_{j,i}} \}$.
	This completes the construction of $G$; see Figure \ref{fig:reduW2} for an example of this reduction. However, for readability purposes, we do not draw all the edges between the vertices in $G'$ and those of $G_j$, for $j \in \{1,2\}$.
	The only such drawn edges are the dotted ones (in gray) that are incident to the vertices $v_1$ and $v_3$.
	We set $\thr = 2\thr^\prime + 1$, $\sol = \thr^\prime$, and $\dom = \{ v_{j,0} \mid j \in \{ 0,1,\ldots,\thr^\prime \} \}$; notice that $\dom$ has $\thr^\prime + 1$ vertices.
	In this way, we constructed the corresponding instance $(G,\thr,\sol,\dom)$.
	Then our claim is that $(G^\prime,\thr^\prime)$ is a \textsf{yes}-instance if and only if $(G,\thr,\sol,\dom)$ has a solution.
	
	\begin{figure}[bt]
		\centering
		\includegraphics[width=0.9\textwidth]{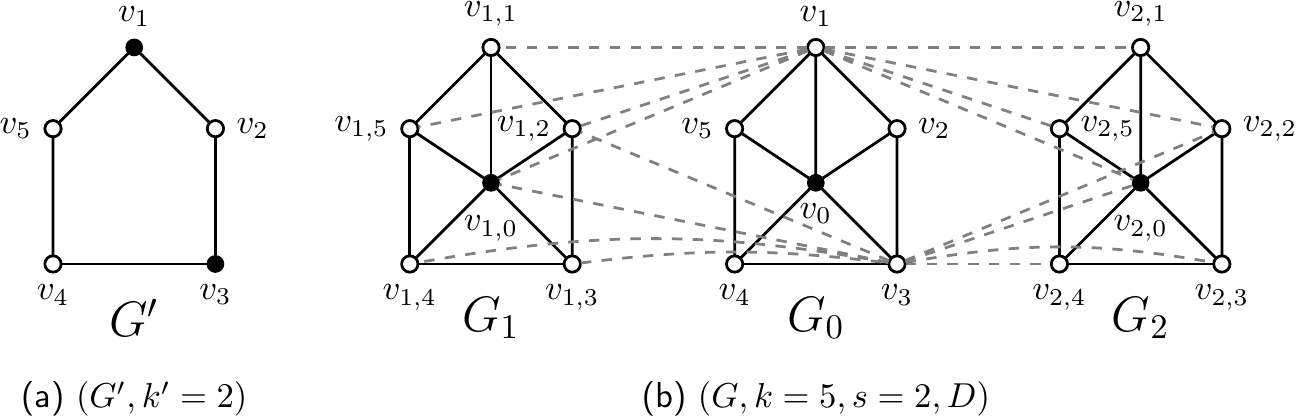}
		\caption{Reduction for Theorem \ref{the:w2-h} with $D' = \set{v_1,v_3}$ and $D = \set{v_0,v_{1,0},v_{2,0}}$.} \label{fig:reduW2}
	\end{figure}
    
    \medskip
    ($\Rightarrow$)
	We first prove the only-if direction.
	Suppose that $(G^\prime,\thr^\prime)$ is a \textsf{yes}-instance, hence there exists a dominating set $\dom^\prime$ of $G^\prime$ of size at most $\thr^\prime$.
	Then by the construction of $G$, we know that $\dom^\prime$ is also a dominating set of $G$ (if we identify the vertices of $G'$ with those of $G_0$).
	Thus it suffices to show that $\dom \sevstepk{\thr} \dom^\prime$, since $\dom^\prime$ has at most $\sol = \thr^\prime$ vertices.
	We first add vertices in $\dom^\prime$ to $\dom$ one by one; this transformation can be done under \TAR($\thr$) since $|\dom \cup \dom^\prime| \leq (\thr^\prime + 1) + \thr^\prime = \thr$.
	We then remove vertices in $\dom$ one by one.
	In this way, we can transform $\dom$ into $\dom^\prime$ under \TAR($\thr$), and hence $(G,\thr,\sol,\dom)$ has a solution $\dom^\prime$.
	
	\medskip
	($\Leftarrow$)
	We then prove the if direction.
	Suppose that $(G,\thr,\sol,\dom)$ has a solution $\dom^\prime$.
	We know that $\dom^\prime$ has at most $\sol = \thr^\prime$ vertices.
	Then, since $G$ has $\thr^\prime + 1$ copies $G_0,G_1,\ldots,G_{\thr^\prime}$, there exists a copy $G_j \in \{ G_0,G_1,\ldots,G_{\thr^\prime} \}$ such that $V(G_j) \cap \dom^\prime = \emptyset$.
	We know that $j \neq 0$ because all neighbors of $v_0$ are in $V(G_0)$, hence $\dom^\prime$ contains at least one vertex in $V(G_0)$.
	For any $p \in \{ 1,2,\ldots, \thr^\prime \} \setminus \{ j \}$, there is no edge joining a vertex in $V(G_j)$ and a vertex in $V(G_p)$.
	Therefore, for any vertex $v_{j,i}$ in $G_j$, a vertex $u \in \dom'$ which dominates $v_{j,i}$ is contained in $(V(G_0) \cap \dom^\prime) \setminus \{ v_0 \}$.
	Then, by the construction of $G$, $u$ also dominates the corresponding vertex $v_i$ in $G_0$.
	Thus, we know that $D^{\prime\prime} = (V(G_0) \cap \dom^\prime) \setminus \{ v_0 \}$ is a dominating set of $G^\prime$.
	Since $|D^{\prime\prime}| \leq |\dom^\prime| \leq \sol=\thr^\prime$ holds, $D^{\prime\prime}$ is a desired dominating set of $G^\prime$.
\end{proof}
    
	On the other hand, we give FPT algorithms with respect to the combination of the solution size $\sol$ and the degeneracy $\dege$ in Subsection~\ref{subsection:d+s} and the vertex cover number $\vc$ in Subsection \ref{sec:fpt_vc}.

	\subsection{FPT algorithm for degeneracy and solution size} \label{subsection:d+s}
		The following is the main theorem in this subsection.
		\begin{theorem} \label{thm:fpt-d+s}
			OPT-DSR is fixed-parameter tractable when parameterized by $\dege + \sol$, where $\dege$ is the degeneracy and $\sol$ the solution size.
		\end{theorem}
		To prove the theorem, we give an FPT algorithm with respect to $\dege + \sol$.
		Note that our algorithm uses the idea of an FPT algorithm solving the reachability variant of {\sc Dominating Set Reconfiguration}, developed by Lokshtanov \emph{et al.}~\cite{LMPRS18}.
		Their algorithm uses the concept of domination core; for a graph $G$, a {\em domination core} of $G$ is a vertex subset $\domcore \subseteq V(G)$ such that any vertex subset $\dom \subseteq V(G)$ is a dominating set of $G$ if and only if $\domcore \subseteq \cneig{G}{D}$~\cite{DFKLPPRVS16}.
		
		Suppose that we are given an instance $(G,\thr,\sol,\dom)$ of OPT-DSR where $G$ is a $\dege$-degenerate graph.
		By Observation~\ref{obs:minimal}, we can assume without loss of generality that $|\dom| < \thr$.
		We first check whether $G$ has a dominating set of size at most $\sol$: this can be done in FPT{\rm (}$\dege + \sol${\rm )} time for $\dege$-degenerate graphs~\cite{AG08}.
		If $G$ does not have it, then we can instantly conclude that this is a \textsf{no}-instance.
		
		In the remainder of this subsection, we assume that $G$ has a dominating set of size at most $\sol$.
		In this case, we kernelize the instance: we shrink $G$ by removing some vertices while keeping the existence of a solution until the size of the graph only depends on $\dege$ and $\sol$.
		To this end, we use the concept of domination core.
		\begin{lemma}[Lokshtanov \emph{et al.}~\cite{LMPRS18}]\label{lem:dom_core}
			If $G$ is a $\dege$-degenerate graph and $G$ has a dominating set of size at most $s$, then $G$ has a domination core of size at most $\dege s^\dege$ and we can find it in FPT$(\dege + s)$ time.
		\end{lemma}
		
		Therefore, one can compute a domination core of $G$ of size at most $\dege \sol^\dege$ in FPT$(\dege + \sol)$ time by Lemma~\ref{lem:dom_core}. 
		In order to shrink $G$, we use the reduction rule \textbf{R1:} if there is a domination core $C$ and two vertices $\rem,\lea \in V(G) \setminus \domcore$ such that $\neig{G}{\rem} \cap \domcore \subseteq \neig{G}{\lea} \cap \domcore$, we remove $v_r$.
		We need to prove that \textbf{R1} is ``safe'', that is, we can remove $\rem$ from $G$ without changing the existence of a solution. However, if the input dominating set $\dom$ contains $\rem$, we cannot do it immediately.
		Therefore, we first remove $v_r$ from $\dom$.
		
		\begin{lemma}\label{lem:vr1}
			Let $\dom$ be a dominating set such that both $|\dom| < \thr$ and $\rem \in \dom$ hold.
			Then there exists $\dom^\prime$ such that $\rem \notin \dom^\prime$ and $\dom \sevstepk{k} \dom^\prime$, and $\dom^\prime$ can be computed in linear time.
		\end{lemma}
	
	    \begin{proof}
	        We first consider the case where $\lea \in \dom$.
	        In this case, we simply remove $\rem$ from $\dom$; let $\dom^\prime$ be the resulting vertex subset.
	        It is clear that $\dom \sevstepk{\thr} \dom^\prime$, and hence it suffices to show that $\dom^\prime$ is a dominating set of $G$.
        	We know that $\domcore \subseteq \cneig{G}{\dom}$ holds by the definition of a domination core.
	        Then since $\neig{G}{\rem} \cap \domcore \subseteq \neig{G}{\lea} \cap \domcore$ and $\lea \in \dom$ hold, we have $\domcore \subseteq \cneig{G}{\dom \setminus \{ \rem \}} = \cneig{G}{\dom^\prime}$.
	        Thus $\dom^\prime$ is a dominating set of $G$.
	
	        We then consider the remaining case where $\lea \notin \dom$.
        	In this case, we can add $\lea$ to $\dom$ since $|\dom| < \thr$. 
        	Then the resulting dominating set contains $\lea$, and we can remove $\rem$ as discussed above. 
        \end{proof}
		
		We can now redefine $\dom$ as a dominating set which does not contain $\rem$.
		We then consider removing $\rem$ from $G$.
		Let $G^\prime = G[V(G) \setminus \{ \rem \}]$.
		The following lemma ensures that removing $\rem$ keeps the existence of a solution.
		
		\begin{lemma}\label{lem:vr2}
			Let $(G,\thr,\sol,\dom)$ be an instance where $\rem \notin \dom$.
			Then, $(G,\thr,\sol,\dom)$ has a solution if and only if $(G^\prime,\thr,\sol,\dom)$ has a solution.
		\end{lemma}
		
		\begin{proof}
		    ($\Leftarrow$)
            We first prove the if direction.
	        Suppose that $(G^\prime,\thr,\sol,\dom)$ has a solution $\dom^\prime_s$.
        	Then there exists a reconfiguration sequence $\mathcal{\dom}^\prime = \langle \dom = \dom^\prime_0,\dom^\prime_1,\ldots,\dom^\prime_{\ell^\prime} = \dom^\prime_\desire \rangle$ of dominating sets of $G^\prime$.
        	It suffices to show that any dominating set $D^\prime_i$ of $G^\prime$ in $\mathcal{\dom}^\prime$ is also a dominating set of $G$.
	        Since $\dom^\prime_i$ is a dominating set of $G^\prime$ and $\rem \notin \domcore$, we have $\domcore \subseteq V(G^\prime) \subseteq \cneig{G^\prime}{\dom^\prime_i}$.
        	By the definition of domination core, we know that $\dom^\prime_i$ is also a dominating set of $G$.
			
			\medskip
			($\Rightarrow$)
	        We then prove only-if direction.
        	Suppose that $(G,\thr,\sol,\dom)$ has a solution $\dom_s$.
        	Then there exists a reconfiguration sequence $\mathcal{\dom} = \langle \dom = \dom_0,\dom_1,\ldots,\dom_\ell = \dom_\desire \rangle$ of dominating sets of $G$.
        	Based on $\mathcal{\dom}$, we construct another sequence $\mathcal{\dom}^\prime = \langle \dom = \dom^\prime_0,\dom^\prime_1,\ldots,\dom^\prime_\ell = \dom^\prime_\desire \rangle$ of vertex sets of $G^\prime$, where
        	\begin{eqnarray*}
	        	\dom^\prime_i =
	    	    \begin{cases}
				    \dom_i \setminus \{ \rem \} \cup \{ \lea \} & (\rem \in \dom_i) \\
				    \dom_i & (otherwise)
			    \end{cases}
		    \end{eqnarray*}
	    for each $i \in \{ 0,1,\ldots,\ell \}$.
        Notice that any vertex subset in $\mathcal{\dom}^\prime$ does not contain $\rem$.
	    Our claim is that $\dom^\prime_s$ is a solution of $(G^\prime,\thr,\sol,\dom)$.
	    To prove it, we show the following two statements: 
			
	    \begin{enumerate}[(i)]
		    \item \label{case:a} for each $i \in \{ 0,1,\ldots,\ell \}$, $\dom^\prime_i$ is a dominating set of $G$ (and hence of $G^\prime$); and
    		\item \label{case:b} for each $i \in \{ 0,1,\ldots,\ell-1 \}$, $|\dom^\prime_i \Delta \dom^\prime_{i+1}| \leq 1$ holds, i.e.\ we have $\dom^\prime_i \onestep \dom^\prime_{i+1}$.
    	\end{enumerate}
			
	    Then the sequence obtained by removing redundant ones from $\mathcal{\dom}^\prime$ is a reconfiguration sequence from $\dom$ to $\dom^\prime_s$.
			
	    We first show the statement \ref{case:a}.
	    Let $\dom_i$ be any dominating set in $\mathcal{\dom}$.
	    If $\rem \notin \dom_i$, then the statement clearly holds.
	    Thus we consider the other case where $\rem \in \dom_i$.
	    Since $\dom_i$ is a dominating set of $G$, we know $\domcore \subseteq V(G) \subseteq \cneig{G}{\dom_i}$.
	    Furthermore, since $\neig{G}{\rem} \cap \domcore \subseteq \neig{G}{\lea} \cap \domcore$, we have $\domcore \subseteq \cneig{G}{\dom_i \setminus \{ \rem \} \cup \{ \lea \}} \subseteq \cneig{G}{\dom^\prime_i}$.
	    By the definition of domination core, $\dom^\prime_i$ is a dominating set of $G$, and hence the statement \ref{case:a} follows.
			
	    We then show the statement \ref{case:b}.
	    Let $\dom_i$ and $\dom_{i+1}$ be any two consecutive dominating sets in $\mathcal{\dom}$.
	    Then, we know $|\dom_i \Delta \dom_{i+1}| = 1$.
	    We assume without loss of generality that $\dom_i \subseteq \dom_{i+1}$; otherwise the proof is symmetric.
	    We prove the statement in the following three cases: 
	    \begin{itemize}
	        \item {\bf Case~1:} both $\rem \notin \dom_i$ and $\rem \notin \dom_{i+1}$ hold; 
	   	    \item {\bf Case~2:} either $\rem \in \dom_i$ or $\rem \in \dom_{i+1}$ holds (but not both); and
	        \item {\bf Case~3:} both $\rem \in \dom_i$ and $\rem \in \dom_{i+1}$ hold.
	    \end{itemize}
			
	    In {\bf Case~1}, we know that $|\dom^\prime_i \Delta \dom^\prime_{i+1}| = |\dom_i \Delta \dom_{i+1}| = 1$, and hence the statement clearly holds.
	    We then consider {\bf Case~2}.
	    In this case, since $\dom_i \subseteq \dom_{i+1}$, we observe that $\rem \notin \dom_i$ and $\rem \in \dom_{i+1}$, and hence $\{ \rem \} = \dom_{i+1} \setminus \dom_{i}$ 
    	Therefore, $\dom^\prime_i \Delta \dom^\prime_{i+1} = \dom_i \Delta (\dom_{i+1} \setminus \{ \rem \} \cup \{ \lea \}) \subseteq \{ \lea \}$.
	    Thus we can conclude that $|\dom^\prime_i \Delta \dom^\prime_{i+1}| \leq 1$, and hence the statement follows.
    	We finally deal with {\bf Case~3}.
    	In this case, we have $\dom^\prime_i \Delta \dom^\prime_{i+1} = (\dom_i \setminus \{ \rem \} \cup \{ \lea \}) \Delta (\dom_{i+1} \setminus \{ \rem \} \cup \{ \lea \}) \subseteq \dom_i \Delta \dom_{i+1}$.
	    Therefore, $|\dom^\prime_i \Delta \dom^\prime_{i+1}| \leq |\dom_i \Delta \dom_{i+1}| = 1$ holds, and hence the statement follows.
	    In this way, we can conclude that $\dom^\prime_s$ is a solution of $(G^\prime,\thr,\sol,\dom)$. This concludes the proof.
        \end{proof}
		
		We exhaustively apply the reduction rule \textbf{R1} to shrink $G$.
		Let $\Gk$ and $\Dk$ be the resulting graph and dominating set, respectively. Then, any two vertices $u,v \in V(\Gk) \setminus \domcore$ satisfy $\neig{\Gk}{u} \cap \domcore \neq \neig{\Gk}{v} \cap \domcore$ (more precisely, $\neig{\Gk}{u} \cap \domcore \not\subseteq \neig{\Gk}{v} \cap \domcore$).
		Then the following lemma completes the proof of Theorem~\ref{thm:fpt-d+s}.
		\begin{lemma}
			$(\Gk,\thr,\sol,\Dk)$ can be solved in FPT$(\dege+\sol)$ time.
		\end{lemma}
		\begin{proof}
			We first show that the size of the vertex set of $\Gk$ is at most $f(\dege,\sol) = \dege \sol^\dege + 2^{\dege \sol^\dege}$.
			Since $|\domcore| \leq \dege \sol^\dege$, it suffices to show that $|V(\Gk) \setminus \domcore| \leq 2^{\dege \sol^\dege}$ holds.
			Recall that any two vertices $u,v \in V(\Gk) \setminus \domcore$ satisfy $\neig{\Gk}{u} \cap \domcore \neq \neig{\Gk}{v} \cap \domcore$.
			Then since the number of combination of vertices in $\domcore$ is at most $2^{|\domcore|} \leq 2^{\dege \sol^\dege}$, we have the desired upper bound $|V(\Gk) \setminus \domcore| \leq 2^{\dege \sol^\dege}$.
			
			We now prove that $(\Gk,\thr,\sol,\Dk)$ can be solved in FPT$(\dege+\sol)$ time.
			To this end, we construct an {\em auxiliary graph} $G_A$, where the vertex set of $G_A$ is the set of all dominating sets of $\Gk$, and any two nodes (that correspond to dominating sets of $\Gk$) $\dom$ and $\dom^\prime$ in $G_A$ are adjacent if and only if $|\dom \Delta \dom^\prime| =1$ holds.
			Let $n = |V(\Gk)|$ and $m = |E(\Gk)|$.
			Then the number of candidate nodes in $G_A$ (vertex subsets of $\Gk$) is bounded by $O(2^{n})$.
			For each candidate, we can check in $O(n+m)$ time if it forms a dominating set.
			Thus we can construct the vertex set of $G_A$ in $O(2^{n}(n+m))$ time.
			We then construct the edge set of $G_A$.
			There are at most $O(|V(G_A)|^2) = O(4^{n})$ pairs of nodes in $G_A$.
			For each pair of nodes, we can check in $O(n)$ time if their corresponding dominating sets differ in exactly one vertex.
			Therefore we can construct the edge set of $G_A$ in $O(4^{n}n)$ time, and hence the total time to construct $G_A$ is $O(4^{n}n + 2^{n}(n+m))$ time.
			We finally search a solution by running a breadth-first search algorithm from $\Dk$ on $G_A$ in $O(|V(G_A)|+|E(G_A)|) = O(4^{n})$ time.
			
			We can conclude that our algorithm runs in time $O(4^{n}n + 2^{n}(n+m))$ in total.
			Since $n \leq f(\dege,\sol)$ and $m \leq n^2 \leq (f(\dege,\sol))^2$, this is an FPT time algorithm.
		\end{proof}
	
	\subsection{FPT algorithm for vertex cover number}\label{sec:fpt_vc}
		Let $(G,k,s,D)$ be an instance of OPT-DSR. As in the previous section, we may first assume by Observation~\ref{obs:minimal} that $\vert D\vert <k$. We recall that $\tau(G)$ is the size of a minimum vertex cover of $G$. In order to lighten notations, we simply denote by $\tau$ the vertex cover number of the input graph. Then, we have the following:
		
		\begin{theorem}\label{thm:fpt-vc}
			OPT-DSR is fixed-parameter tractable when parameterized by $\tau$.
		\end{theorem}
	
		We first establish the following fact that is going to be useful later.
		
		\begin{observation} \label{obs:degeneracy}
			If $G$ is $d$-degenerate, then $d \leq \tau$.
		\end{observation}
		
		\begin{proof}
			Let $G$ be a graph, $X$ a minimum vertex cover of $G$ and $H$ be any subgraph of $G$. Recall that $G[V \setminus X]$ is an independent set. If $H$ contains a vertex $v$ outside $X$, then $v$ has a degree at most $\tau$ in $G$ and therefore in $H$. Otherwise, $H$ is a subraph of $G[X]$ and thus has at most $\tau$ vertices. Hence all vertices of $H$ have degree at most $\tau$ in $H$. Therefore, since any subgraph $H$ of $G$ contains a vertex of degree at most $\tau$, $G$ is $\tau$-degenerate.
		\end{proof}

		We are now able to get down to the proof of Theorem~\ref{thm:fpt-vc}, by providing an algorithm that solves OPT-DSR and runs in time FPT$(\tau)$.
		We first compute a minimum vertex cover $X \subseteq V(G)$ of $G$ in time FPT($\tau$) \cite{CHEN20103736}. We partition the vertices of $G$ into two components, the vertex cover $X$ and the remaining vertices $I$. By definition of vertex cover, no edge can have both endpoints outside $X$, therefore $I$ is an independent set.
		Note that if $s \le \tau$, then by Observation~\ref{obs:degeneracy} we have $d+s \leq 2\tau$, where $d$ is the degeneracy of $G$. In this case we are able to use the algorithm of the last section, that runs in time FPT$(d+s)$.
		
		We may therefore assume $\tau < s$. In that case, we have the following lemma:
	
		\begin{lemma} \label{lemma:FPTyes}
		If $\tau < s$, then $(G,k,s,D)$ is a {\sf yes}-instance.
		\end{lemma}
		
		\begin{proof}
	In the remainder of the proof, we assume that the graph $G$ has no isolated vertex since an isolated vertex \emph{must} belong to any dominating set of $G$. We now prove that $(G,k,s,D)$ is always a \textsf{yes}-instance, i.e.\ there exists a dominating set of size at most $\tau$ that is reachable from $D$ under the TAR($k$) rule. 
		
	We associate to every vertex $v \in X\setminus D$ a \textit{special neighbor} among its neighbors that dominate it (which can be either in $X$ or $I$), i.e.\ we pick arbitrarily a vertex in $N_G[v] \cap D$. We denote this special neighbor $t(v)$. Let $T$ be the set of special neighbors, i.e.\ $T := \{ t(v)\ \vert\ v\in X\setminus D\}$. This corresponds to the set of vertices that are used to dominate the vertices in $X$ that do not belong to $D$. Note that $|T| \le \tau$.
		
    We are now able to describe the algorithm we use to output $D_\desire$, the target dominating set. It consists in exhaustively applying the two following rules on the vertices of $I$ that belong to the current dominating set:
		
	\begin{enumerate}[(i)]
		\item \label{case:vc-a} if there is a vertex $v$ in $I$ but not in $T$ that is already dominated by another vertex, then we remove $v$ from the dominating set; and
		\item \label{case:vc-b} if there is a vertex $v$ in $I$ but not in $T$ that is dominated only by itself, then we add any one of its neighbors $u \in X$ to the dominating set, and then remove $v$. The vertex $u$ does not need a special neighbor anymore, since it now belongs to the dominating set. We thus update the set $T$ by only keeping the special neighbors $t(w)$ of vertices $w$ that are still in $X \setminus D$.
	\end{enumerate}
		
	\begin{figure}[bt]
	    \centering
        \includegraphics[width=0.9\textwidth]{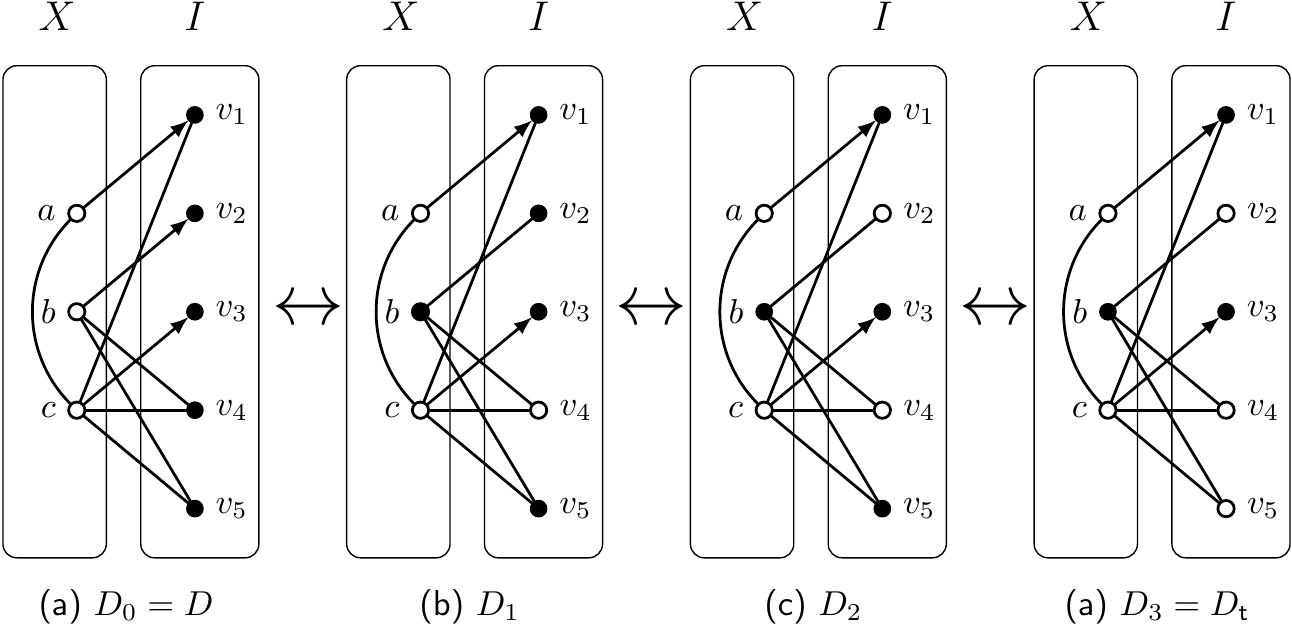}
		\caption{Reconfiguration sequence from the original dominating set $D = I$ to the target one $D_\desire = \set{b,v_1,v_3}$. $D_1$ is obtained from $D_0$ by applying Rule \ref{case:vc-b} and $D_2$ (resp.\ $D_3$) obtained from $D_1$ (resp.\ $D_2$) by applying Rule \ref{case:vc-a}. The special neighbor of a vertex $v \in X \setminus D$ is the one pointed by its outgoing edge.}
		\label{fig:fpt-vc}
	\end{figure}
		
	We first prove that these two rules are safe, i.e.\ we do not break the domination property at any step. Since Rule \ref{case:vc-a} removes a vertex $v$ that is not required to dominate itself or another vertex $u \in X$ (because it has not been chosen in $T$), we can safely remove it. In Rule \ref{case:vc-b}, after adding a neighbor of $v$ to the dominating set, $v$ is not required to dominate itself anymore. Since $v$ is not in $T$, we can now apply Rule \ref{case:vc-a} which is safe.
		
	Recall that $|D| < k$. Then, each dominating set obtained after applying one of these rules is of size at most $k$ since Rule \ref{case:vc-a} only removes vertices and Rule \ref{case:vc-b} consists in an addition immediately followed by a removal. 
		
	Now, let $D_\desire$ be the dominating set obtained once we cannot apply Rule \ref{case:vc-a} and Rule \ref{case:vc-b} anymore (see Figure \ref{fig:fpt-vc} for an example). All remaining vertices in $I \cap D_\desire$ now belong to $T$. By definition of $T$, each vertex in $X \setminus D_\desire$ has (exactly) one neighbor in $T$ (but they are not necessarily distinct). Therefore, $|I \cap D_\desire| \le |X \setminus D_\desire|$. As a result, $|D_\desire| = |X \cap D_\desire| + |I \cap D_\desire| \le |X \cap D_\desire| + |X \setminus D_\desire| = |X| = \tau$. Since $\tau < s$, the size of $D_\desire$ is at most $s$, as desired. 
    \end{proof}
	
	It remains to discuss the complexity of this algorithm. As we already said, we first compute a minimum vertex cover $X$ of $G$ in time FPT($\tau$). If $s \le \tau$, we run the FPT algorithm of Section \ref{subsection:d+s}. Otherwise, we first compute the set $T$ and then run the subroutine which are both described in the proof of Lemma \ref{lemma:FPTyes}. The two rules used in this subroutine only apply to vertices that belong to the set $I$ and whenever one is applied, exactly one vertex in $I$ is removed (and none is added). Hence, they are applied at most $|I \cap D|$ times. Therefore, the subroutine runs in polynomial time and produces the desired dominating set $D_\desire$. As a result, this algorithm is FPT with respect to $\tau$. This concludes the proof.		
	
\paragraph{Concluding remarks.} In this paper, we showed that {\sc OPT-DSR} is PSPACE-complete even if restricted to some graph classes. However, we only know that it is NP-hard for bounded maximum degree graphs or planar grapĥs, as an immediate corollary of Observation~\ref{obs:equivalent}. Hence, it would be interesting to determine whether {\sc OPT-DSR} is NP-complete or PSPACE-complete on these two graph classes. Note that the complexity on planar graphs remains open for {\sc OPT-ISR}.

We also proved that {\sc OPT-DSR} is W[2]-hard for parameter $k$ but the question remains as to whether there exists an XP algorithm for upper bound $k$.

\paragraph{Acknowledgements.}
We would like to thank the anonymous referees for several remarks which helped improve the presentation of this paper.

%
%
\bibliographystyle{abbrv}
\bibliography{references}

\end{document}